\documentclass[journal]{IEEEtran}
\usepackage{cite}
\usepackage{amsmath,amssymb,amsfonts}
\usepackage{amsthm}
\newtheorem{theorem}{Theorem}
\newtheorem{example}{Example}
\newtheorem{definition}{Definition}
\newtheorem{proposition}{Proposition}
\newtheorem{remark}{Remark}
\usepackage{algorithm,algorithmic}                                                                                                              
\usepackage{graphicx}
\usepackage{textcomp}
\usepackage{xcolor}
\def\BibTeX{{\rm B\kern-.05em{\sc i\kern-.025em b}\kern-.08em
    T\kern-.1667em\lower.7ex\hbox{E}\kern-.125emX}}
\ifCLASSINFOpdf

\else
\fi
\begin{document}

\title{Signal Recovery from Time and Frequency Samples} 
\author{\IEEEauthorblockN{Mert Kayaalp and Oleg Szehr}
\thanks{The authors are affiliated with the Dalle Molle Institute for Artificial Intelligence (IDSIA) - SUPSI/USI, Via la Santa 1, 6962 Lugano-Viganello, Switzerland. Emails: \{mert.kayaalp, oleg.szehr\}@idsia.ch. 
}
}

\markboth{Preprint}%
{Kayaalp and Szehr}

\maketitle

\begin{abstract}
We analyze signal recovery when samples are taken concomitantly from a signal and its Fourier transform. This two-sided sampling framework extends classical one-sided reconstruction and is particularly useful when measurements in either domain alone are insufficient because of sensing, storage, or bandwidth constraints. We formulate the resulting recovery problem in finite-dimensional spaces and reproducing kernel Hilbert spaces, and illustrate the infinite-dimensional setting in a Fourier-symmetric Sobolev space. Numerical experiments with sinc- and Hermite-based schemes indicate that, under a fixed sampling budget, two-sided sampling often yields better conditioned systems than one-sided approaches. A simplified spectrum-monitoring example further demonstrates improved reconstruction when limited time samples are supplemented with frequency-domain information.
\end{abstract}


\IEEEpeerreviewmaketitle

\section{Introduction}

\IEEEPARstart{S}{ampling} and reconstruction are central to analog-to-digital conversion, interpolation, and resampling operations needed in signal processing. The classical Whittaker-Shannon-Nyquist theorem \cite{whittaker1915xviii} guarantees unique and stable recovery of bandlimited signals from sufficiently dense time-domain samples \cite{landau1967,seip2004interpolation}. However, this framework - and most of its extensions - rely on \lq\lq{}one-sided\rq\rq{} measurements, using either time or frequency samples alone.

Recent mathematical advances have established a theory for function interpolation from the values of the function and its Fourier transform~\cite{radchenko2019fourier,Kulikov2021,kulikov2023fourier,szehr2025spectral}. This work adopts an engineering perspective, and adds applications and implementable two-sided sampling schemes that leverage simultaneous time- and frequency-domain samples for practical signal recovery.

In practice, two-sided sampling is most useful when neither domain alone provides sufficient resolution for unique recovery, e.g.~due to sensing, storage, or bandwidth constraints. Applications include spectrum monitoring, radar, wireless communication, and MRI (see Sec.~\ref{sec:motivatingExamples}). Here, limited time-domain data can be supplemented with frequency-domain measurements - either direct Fourier samples, coarse spectral summaries, or a few targeted frequency observations - yielding additional information for reconstruction. This motivates the central question of our work:

\vspace{0.5em}
\noindent \fbox{
\centering
\parbox{0.90\linewidth}{\textit{How can samples from both time and frequency domains be leveraged for practical signal recovery?}
}} \\

We provide a mathematical framework for reconstruction (see Secs.~\ref{sec:reconstruction_inFiniteDimensions} and~\ref{sec:reoncstruction_inINfiniteDimensions}) and numerical evidence that reconstructions with two-sided data improves identifiability (see Secs.~\ref{sec:numerical_exps} and~\ref{sec:spectrum_monitoring}).

\subsection{Motivation}\label{sec:motivatingExamples}

Joint use of time and frequency-domain samples can improve existing sampling architectures and help overcome important practical limitations. 

\textit{Spectrum monitoring \& memory limitations:} 
In many applications the full set of time-domain samples cannot be stored due to memory constraints. This situation is common in {spectrum monitoring}, where storing all digitized time samples can be prohibitive. Thus systems retain only limited time samples after conversion to spectral information \cite{axell2012spectrum,spectrumLunden2015,itur_sm2117_2018}. 
We provide a simplified numerical example for this in Sec.~\ref{sec:spectrum_monitoring}, where supplementing the stored time-information with frequency-domain information improves the reconstruction of the underlying time signal compared to using stored time samples only.

\emph{MRI \& Prior information:}
In some applications, measurements are acquired in one domain, while samples in the other domain are available as a prior information. For example, in \textit{magnetic resonance imaging (MRI)}, the acquired measurements are inherently obtained in the Fourier domain ($k$-space).
At the same time, in some settings, a subset of the spatial-domain samples of the target image may be known in advance (e.g., known background pixels, prior anatomical structure, fiducial markers, or trusted auxiliary/reference data)~\cite{liang1994efficient,zhi1996constrained,weizman2016reference}. 
Such spatial information can be combined with the acquired $k$-space data to improve the reconstruction.

\emph{Communications \& pre-designed signal supports:} 
In many applications, the signals' supports follow a regular structure. For instance, in \emph{radar or communication systems}, transmitted signals often lie on a pre-designed time-frequency grid. In multicarrier systems such as orthogonal frequency-division multiplexing (\textit{OFDM})~\cite{bahai2004multi}, the signals are supported on specific grids defined by global communication standards such as 4G~LTE and 5G~NR \cite{3gpp_ts38211}. 
Since these grids are known, receivers can allocate a small number of parallel analog front-end branches to the corresponding frequencies and acquire narrowband spectral measurements in addition to time-domain samples. 
Such measurements can be obtained with standard RF hardware, for example by using a small analog filter-bank/channelizer or coherent correlator branches tuned to the desired frequencies.
This can become even more valuable in scenarios where frequency supports may shift slightly due to Doppler effects in \textit{high-mobility settings} or due to \textit{hardware imperfections}~\cite{rappaport2010wireless} such as clock-offset, skew and phase noise. 

\subsection{Background \& Related Literature}

The reconstruction of a function from its values on a discrete set is a fundamental problem in analysis. In classical sampling theory, the Whittaker–Shannon-Nyquist theorem asserts that band-limited functions are uniquely determined by samples taken at a rate beyond the Nyquist threshold. Beyond the band-limited setting, this role is assumed by the notion of a \emph{uniqueness set}: a set $\Lambda \subset \mathbb{R}$ such that any function in a given space $\mathcal{X}$ vanishing on $\Lambda$ must vanish identically. The characterization of uniqueness sets is closely tied to the geometry of $\mathcal{X}$ and to quantitative density conditions on $\Lambda$~\cite{seip2004interpolation}.

More recently, in the context of Fourier interpolation and crystalline measures (see~\cite{radchenko2019fourier, ramos2022fourier,kulikov2023fourier,szehr2025spectral} and related works), the notion of a \emph{uniqueness pair} - a pair $(\Lambda,M)$ consisting of a spatial set $\Lambda$ and a frequency set $M$ - has been introduced.
\begin{definition}[Uniqueness pair, e.g.~\cite{kulikov2023fourier}]
    Let $\mathcal{X}$ be a space of functions on $\mathbb{R}$. A pair of sets $(\Lambda, M)$, $\Lambda,M\subset \mathbb{R}$ is a uniqueness pair for $\mathcal{X}$ if for $f\in\mathcal{X}$ the conditions $f|_\Lambda=0$ and $\hat{f}|_M=0$ imply that $f=0$.
\end{definition}
By linearity, $(\Lambda,M)$ is a uniqueness pair for $\mathcal{X}$ if prescribing $f|_{\Lambda}$ and $\hat{f}|_{M}$ determines $f \in \mathcal{X}$ uniquely. Such uniqueness typically requires the sampling densities of $\Lambda$ and $M$ to exceed appropriate critical thresholds. 

To date, contributions in this area have been purely mathematical in nature. The pioneering work or Radchenko-Viazovska~\cite{radchenko2019fourier} introduces mixed Fourier interpolation via modular forms. Subsequent studies~\cite{ramos2022fourier,kulikov2023fourier} determine density conditions for such schemes. They show, in particular, that the construction of~\cite{radchenko2019fourier} lies at a critical threshold consistent with the $2WT/\pi$-barrier of Landau-Pollak-Slepian~\cite{Slepian1961,Slepian1961b,Landau1962}, see~\cite{Kulikov2021}, and also the extremal density condition of~\cite{kulikov2023fourier}. The work \cite{szehr2025spectral} extends the uniqueness theory to transforms beyond the Fourier setting. These mathematical works stop short of offering implementable reconstruction\footnote{The proposed interpolation formulas require sampling densities that increase with time and frequency making a direct translation into practical recovery schemes difficult.} algorithms or recipes accessible to signal processing practice.

From an applied perspective, many systems possess samples in both the time (or spatial) and frequency domains. Yet standard reconstruction pipelines typically treat these data separately, and neglect the intrinsic coupling between a function and its Fourier transform. 
Related in spirit are Papoulis-Gerchberg-type iterative schemes \cite{gerchberg1974super,papoulis1975,hirani1996combining,matsuki2009spectroscopy}, which enforce function space prior constraints such as bandlimitedness.  
They can be interpreted as alternating projections between measurement consistency in one domain and structural constraints in the transform domain \cite{ferreira2002interpolation,marques2006papoulis}. In contrast, we develop a framework for reconstruction from concomitant spatial and Fourier samples that directly exploits the joint information contained in both domains.

This setting should also be distinguished from classical localized time-frequency representations such as wavelet methods.
Those approaches analyze a signal through its correlations with shifted or modulated atoms and produce coefficients indexed by a time-frequency plane. 
In the present work, the data instead consist of direct samples \(f(t_i)\) together with \textit{global} Fourier samples \(\widehat f(\omega_j)\). 
The resulting problem is therefore a mixed time and Fourier interpolation rather than the inversion of a localized time-frequency representation. 
Likewise, the framework is distinct from compressed sensing since our setting imposes no sparsity condition.

Our two-sided approach has several consequences. First, it enables unique recovery beyond the bandlimited setting. Second, when hardware constraints preclude sampling or storing at the Nyquist rate, additional Fourier measurements can significantly enhance reconstruction quality. Third, our numerical experiments indicate that, in finite-dimensional settings, the resulting two-sided linear systems are often better conditioned than their one-sided counterparts under equal sample budgets.

\subsection{Notation}\label{sec:notation}
We adopt the unitary convention for the Fourier transform, denoting time domain variables by $t$ and frequency domain variables by $\omega$,
\begin{align*}
	\hat{f}(\omega) &= {F}[f](\omega) = \frac{1}{\sqrt{2\pi}}\int_{\mathbb{R}}f(t)e^{-i\omega t}\textnormal{d} t,
\end{align*}
with $f(t) =  {F}^*[\hat f(t)]$ for all suitable $f$.

We write $\mathcal{X}$ for a vector space of functions on $\mathbb{R}$. In our exposition, we will focus on two guiding examples of function spaces. The first is the classical Paley-Wiener space of bandlimited functions, 
\begin{equation*}
    \mathcal{PW}_{\Omega}(\mathbb{R}) = \left\{ f \in L^2(\mathbb{R}) : \hat{f}(\omega) = 0 \text{ for } |\omega| > \Omega \right\}.
\end{equation*}
The second example is the Fourier-symmetric Sobolev space (aka.~modulation space), which is defined as~\cite{ehler2024abstract} (see also~\cite{kulikov2023fourier,zelent2025time,szehr2025spectral})
 \begin{equation*}
         \mathcal H \!=\!\!\Big\{f\in L^2(\mathbb R)\ \!\!: \! \! 
    \int_{\mathbb R} \!\!t^2|f(t)|^2\,dt + \! \! \int_{\mathbb R} \!\!\omega^2|\hat f(\omega)|^2\,d\omega < \! \infty\Big\},
 \end{equation*}
where the first term represents the time-domain regularity and the second term reflects the frequency-domain regularity. Due to the symmetry of these terms, the space is invariant under the Fourier transform and therefore constitutes a natural framework for the study of uniqueness pairs. Furthermore, it admits an orthogonal basis of Hermite functions, which are eigenfunctions of the Fourier transform, see below.

We will treat these spaces as spanned by orthogonal bases or as Reproducing Kernel Hilbert Spaces (RKHS). Recall that an RKHS $\mathcal{X}$ of functions on a set \( X \) is defined by the requirement that the point evaluation functional \( f \mapsto f(x) \) is continuous. By the Riesz representation theorem, this is equivalent to the existence of a unique function $K_x\in\mathcal X$ with the reproducing property $f(x)=\langle f,K_x\rangle_{\mathcal X}$, and it induces a positive semidefinite reproducing kernel  \( K: X \times X \to \mathbb{C} \), $K(x,y):=\langle K_y,K_x\rangle_{\mathcal X}$, see~\cite{Aronszajn1950} and~\cite{unser2000} for applications to sampling.

\section{Contributions}
\begin{itemize}
  \item \textit{Formalization of two-sided time-frequency sampling.}
  We formulate reconstruction from measurements consisting of samples of both $f$ in time and samples of $\hat f$ in frequency domains. We describe basis sampling with expansions of $f$ in suitable reconstruction families $\{\Phi_n\}$. For reconstruction, this yields a stacked linear system whose blocks enforce constraints in both domains as a coefficient vector solving problem.
  \item \textit{Two-sided sampling in finite-dimensional spaces.}
   We investigate two-sided sampling in function spaces spanned by a finite family of $N+1$ basis functions. In this finite-dimensional setting, unique signal recovery is equivalent to the existence of a unique solution in a $(N+1)\times (N+1)$ linear system whose entries consist of evaluations of the basis functions and their Fourier transforms at the prescribed time and frequency samples. Commonly used bases in signal processing include shifted sinc functions, Hermite functions, and prolate spheroidal wave functions. We demonstrate the resulting recovery schemes through numerical experiments in this framework. 
   \item \textit{Two-sided sampling in reproducing kernel Hilbert spaces.}
We introduce reconstruction from  two-sided samples within an RKHS framework and derive an explicit kernel-based representation for mixed constraints of $f$ and $\hat f$ on sampling sets $\Lambda$ and $M$, respectively. 

  \item \textit{Two-sided sampling in Fourier-symmetric Sobolev space.} As an illustrative example in an infinite-dimensional setting, we consider sampling in the Fourier-symmetric Sobolev space $\mathcal{H}$, which has been identified as a suitable object for the study of two-sided interpolation in the mathematical literature. The space $\mathcal{H}$ naturally arises in contexts such as frequency modulation~\cite{ehler2024abstract}, and, unlike the Paley-Wiener space, it allows for the study of settings that go beyond the band-limitation constraint.

\item  \textit{Numerical comparison between one- and two-sided sampling:}
We perform numerical experiments to compare reconstruction from one-sided and two-sided samples under a fixed sampling budget. Specifically, we analyze the condition numbers of the resulting finite-dimensional linear systems arising in sinc- and Hermite-based sampling schemes. Our results show that two-sided systems typically exhibit smaller condition numbers, which indicates reduced sensitivity to noise, improved numerical stability of the recovery algorithms, and enhanced robustness of reconstruction.

\item \textit{Spectrum monitoring:} 
As a practical case study, we evaluate a simplified spectrum monitoring application in Sec.~\ref{sec:spectrum_monitoring}, where the system is constrained in the number of time-domain samples it can store. 
We demonstrate that reconstruction performance improves when the available time-domain samples are supplemented with measurements from monitored frequency bins.
\end{itemize}

\section{Reconstruction from Two-Sided Samples}\label{sec:reconstruction_inFiniteDimensions}
The classical Whittaker-Shannon-Nyquist sampling theorem asserts that if a function \(f\) is bandlimited to \(|\omega|\le \pi/T\)  (i.e.~$f\in \mathcal{PW}_{\pi/T}$) then it admits an expansion by taking samples at a sampling step of $T$,
\[
f(t)=\sum_{n\in\mathbb Z} f(nT)\,\operatorname{sinc}\!\left(\frac{t-nT}{T}\right),\ \operatorname{sinc}(x)=\frac{\sin(\pi x)}{\pi x}.
\]
This formula can be viewed simultaneously as a prototypical instance of sampling in an orthonormal basis (the “basis-sampling” picture) and as sampling by point evaluations in a reproducing kernel Hilbert space (the “RKHS” picture)~\cite{unser2000}. In the former picture $f$ is expanded as
\begin{equation}\label{eq:basis_time_expansion}
    f(t)=\sum_{n\in\mathbb Z}\alpha_n\,\Phi_n(t),
\end{equation}
with $\Phi_n(t)=\sqrt{\dfrac{1}{T}}\;\operatorname{sinc}\!\left(\dfrac{t-nT}{T}\right)$. Direct computation reveals that $\Phi_n$ are orthonormal
\begin{align*}
\langle \Phi_n,\Phi_m\rangle&=\int_{-\infty}^{\infty}\Phi_n(t)\,\Phi_m(t)\,dt=\delta_{nm},\ \textnormal{with}\\
\alpha_n&=\langle f,\Phi_n\rangle=\sqrt{T}\,f(nT).
\end{align*}
In the RKHS picture $f$ is expanded in reproducing kernels evaluated at sampling points
\begin{align}\label{eq:RKHSExpansion}
f(t)=\sum_{n\in\mathbb{Z}}\alpha_n K_{t_n}(t)
\end{align}
with
\[
K(x,y)=K_y(x)=\frac{\sin(\pi/T(x-y))}{\pi(x-y)} = \frac 1 T \operatorname{sinc}\left(\dfrac{x-y}{T}\right)
\]
and $t_n=nT$. The expansion formulas~\eqref{eq:basis_time_expansion} and~\eqref{eq:RKHSExpansion} constitute fundamental components of sampling theory and find widespread application beyond the classical bandlimited framework. It is worth noting that in such settings the corresponding kernels need not be orthogonal.

If data or hardware constraints preclude sampling at a sufficiently high rate, aliasing occurs and unique signal reconstruction is no longer possible. In such cases, it is natural to attempt for incorporating additional information, such as measurements obtained in the Fourier domain. We begin by illustrating reconstruction from two-sided sampling in a finite-dimensional setting. To this end, we introduce a sample budget (which roughly corresponds to a maximal sampling frequency) and study signal recovery subject to this constraint in both basis sampling and RKHS sampling scenarios.

\subsection{Finite basis sampling with two-sided samples}\label{sec:finite_basis_sampling}
Assuming the signal lies in a finite vector space $$\mathcal{X}_N=\text{span}\{\Phi_0,....,\Phi_N\},$$ basis sampling represents a signal by expanding it in a given family of reconstruction functions \(\{\Phi_n\}_{n=0}^N\), with coefficients \(\{\alpha_n\}_{n=0}^N\), cf~\eqref{eq:basis_time_expansion}. The coefficients $\{\alpha_n\}_{n=0}^N$ serve as the discrete signal representations for digital processing.

We may also expand the Fourier transform in the corresponding Fourier basis. In particular, taking the Fourier transform (and assuming sufficient regularity) of \eqref{eq:basis_time_expansion} yields
\begin{equation*}
    \hat f (\omega) = \sum_{n=0}^N\alpha_n\, \hat \Phi_n (\omega).
\end{equation*}
If $K$ samples can be obtained from \(f(t)\) and $L$ samples can be obtained from \(\hat f(\omega)\) this yields $K$ linear equations that constrain $\{\alpha_n\}_{n=0}^N$ through time-domain information and $L$ linear equations for the frequency-domain. Combining the two, one can attempt to reconstruct \(f\) by solving a stacked linear system. Specifically, writing \(\vec c\) for the measurements obtained from \(f(t)\) and  \(\vec{\hat c}\) for the measurements obtained from \(\hat f(\omega)\) this yields the block system
\begin{equation}\label{eq:general_basis_inversion}
    \begin{pmatrix}
        \vec c\\[1mm]
        \vec{\hat c}
    \end{pmatrix}
    =
    \begin{pmatrix}
       \left(\Phi_{ij}\right) \\[1mm]
        \left({\hat \Phi}_{ij}\right)
    \end{pmatrix}
    \begin{pmatrix}
        \vec\alpha
    \end{pmatrix}.
\end{equation}
Here we write $\left(\Phi_{ij}\right)$ for the matrix with entries $\Phi_{ij}=\Phi_j(t_i)$, $0\leq j\leq N$, $0\leq i\leq K-1$ and $\left(\hat\Phi_{ij}\right)$ for the matrix with entries \(\hat\Phi_{ij}=\hat\Phi_j(\omega_i)\), $0\leq j\leq N$, $0\leq i\leq L-1$ and $\vec\alpha$ for the coefficient vector of $\{\alpha_n\}_{n=0}^N$.
The solution set may be empty (no solution, i.e., inconsistency), a singleton (exactly one solution), or non-unique (infinitely many solutions), depending on the system's rank conditions and consistency. Assuming  $f\in\mathcal{X}_N$ and exact measurements guarantees the existence of a solution, but  uniqueness necessitates $K+L\geq N+1$ and is not guaranteed in general. 
Systems of the form~\eqref{eq:general_basis_inversion} can be treated uniformly via the Moore--Penrose pseudoinverse $A^{\dagger}$. For $Ax=b$, the pseudoinverse yields $x^\star = A^{\dagger} b$, which is a least-squares solution when the system is overdetermined (minimizing $\|Ax-b\|_2$); when multiple least-squares solutions exist, it selects the minimum-norm one. In the underdetermined case, it returns the minimum-norm solution among all exact solutions when $b\in \mathrm{range}(A)$, and otherwise the minimum-norm least-squares solution. The pseudoinverse arises most transparently from the singular value decomposition: if $A=U\Sigma V^\dagger$ (with singular values $\sigma_i$ on the diagonal of $\Sigma$), then
\[
A^{\dagger}=V\Sigma^{\dagger}U^\dagger,\qquad
\Sigma^{\dagger}=\mathrm{diag}\!\left(\sigma_i^{-1}\;|\;\sigma_i>0\right),
\]
with zeros left in place of any vanishing singular values. Thus the singular values determine both existence/uniqueness through $\mathrm{rank}(A)$ and numerical behavior: small $\sigma_i$ correspond to weakly constrained directions and lead to ill-conditioned behavior in the sense of amplification in $A^{\dagger}b$. We study conditioning experimentally in Sec.~\ref{sec:numerical_exps}. The above procedure is summarized in Algorithm~\ref{alg:finite_basis}.

\begin{algorithm}[]
\caption{Finite basis reconstruction from two-sided samples}
\label{alg:finite_basis}
\begin{algorithmic}[1]
\REQUIRE Basis functions $\{\Phi_n\}_{n=0}^{N}$; sampling points in time  $\{t_i\}_{i=0}^{K-1}$ and frequency $\{\omega_i\}_{i=0}^{L-1}$; measurements $\vec c$, $\vec{\hat c}$. 
\STATE Form the time block $A_t=(\Phi_j(t_i))_{i,j}$.
\STATE Form the frequency block $A_\omega=(\hat{\Phi}_j(\omega_i))_{i,j}$.
\STATE Set $A=[\begin{smallmatrix}A_t\\ A_\omega\end{smallmatrix}]$ and $b=(\vec c,\vec{\hat c})^\top$.
\STATE Compute $\vec\alpha^\star=A^\dagger b$
\ENSURE Reconstructed function $f^\star=\sum_{n=0}^{N}\alpha_n^\star\Phi_n$.
\end{algorithmic}
\end{algorithm}

The rigorous study of uniqueness quickly leads into deep mathematics. A polynomial of degree $N$ is uniquely determined by its values at $N+1$ distinct nodes. Consequently, if $\mathcal{X}_N$ is (essentially) a polynomial space, any non-degenerate measurements yield a (one-sided) uniqueness set. Yet the corresponding two-sided linear system may  admit multiple solutions. Thus, even in finite dimension and in the polynomial setting, a uniqueness set for pointwise interpolation does not automatically induce a unique solution for the stacked system. We illustrate this point with the following example.

\subsubsection{The finite Hermite function space}
The Hermite functions are defined by the relation
$$\varphi_n(x)=\frac{1}{\sqrt{n!}}\left(\frac{1}{\sqrt{2}}(x-\partial_x)\right)^n\varphi_0(x),\ n\geq1,$$ with $\varphi_0=\pi^{-1/4}e^{-x^2/2}$. It is obvious that every Hermite function has the form $\varphi_n(x)= \frac{1}{\pi^{1/4}\sqrt{2^nn!}}H_n(x)e^{-x^2/2}$ with a degree-$n$ polynomial $H_n$. These polynomials are the Hermite polynomials. We consider the $(N+1)$-dimensional Hermite function space
\[
\mathcal{H}_N=\mathrm{span}\{\varphi_0,\dots,\varphi_{N}\},
\]
where any $f\in \mathcal{H}_N$ admits the unique expansion $$f(t)=\sum_{n=0}^{N} \alpha_n\,\varphi_n(t).$$
The Hermite functions are eigenfunctions of the Fourier transform
\[
F[\varphi_n](w)=(-i)^n\,\varphi_n(\omega),\qquad n\ge 0.
\]
Therefore for $f\in \mathcal{H}_N$, it holds that
\[
\hat f(\omega)=\sum_{n=0}^{N} \alpha_n\,(-i)^n\,\varphi_n(\omega).
\]
\paragraph{One-sided sampling}
Pure time or frequency sampling at $N+1$ distinct nodes allows for unique reconstruction. If measurements are taken from either time or frequency domain, then the system of equations~\eqref{eq:general_basis_inversion} only contains
$\Phi_{i,j}=\varphi_{j}(t_i)$ or $\hat\Phi_{i,j}=(-i)^j\varphi_{j}(\omega_i)$. Making use of the Hermite expansion, any $f\in \mathcal{H}_N$ can be written in the form $f(x)=e^{-x^2/2}\,p(x)$ with a polynomial $p$ of degree $N$. Prescribing $N+1$ zeros $f(t_j)=0$ at distinct points in time implies that the degree $N$ polynomial $p$ has $N+1$ distinct zeros $p(t_j)=0$. Thus $p=0$ and $f=0$. In other words any distinct points $\{t_0<t_1<....<t_N\}\subset\mathbb{R}$ constitute a uniqueness set for $\mathcal{H}_N$. Along the lines of Proposition~\ref{prop:uniqueness_pair}, a unique signal interpolates $N+1$ distinct sampling points. The same reasoning applies also for one-sided frequency domain measurements.
\paragraph{Two-sided sampling}
The situation is more delicate in the case of two-sided sampling. Choosing $K$ distinct measurements from the time domain and $L$ distinct measurements from the frequency domain with $K+L=N+1$ is not guaranteed to yield a unique solution. We illustrate this with a simple example:
\begin{example}\label{examp:nonunique_f}
Consider the space $\mathcal{H}_2$, where $N+1=3$ measurements are taken in total of which $K=1$ measurements are taken from the time domain and $L=2$ from the frequency domain: 
\[
t_0=0,\qquad \omega_0=1,\qquad \omega_1=-1.
\]
In this space, consider the function: 
\[
f=\frac{1}{\sqrt2}\varphi_0+\varphi_2 \in \mathcal{H}_2
\]
with $\hat{\varphi}_0=\varphi_0$ and $\hat{\varphi}_2=(-i)^2\varphi_2=-\varphi_2$, so that
\[
\hat f=\frac{1}{\sqrt2}\varphi_0-\varphi_2\in \mathcal{H}_2.
\]
We have
\[
f(t_0=0)=\frac{1}{\sqrt2}\varphi_0(0)+\varphi_2(0)
=\frac{1}{\sqrt2}(\pi^{-1/4}-\pi^{-1/4})=0.
\]
Since $\varphi_0$ and $\varphi_2$ are even and
\[
\varphi_0(\pm1)=\pi^{-1/4}e^{-1/2},
\quad
\varphi_2(\pm1)=\frac{1}{\sqrt2}\pi^{-1/4}e^{-1/2}
\]
we find
\[
\begin{aligned}
    \hat f(\omega_{0/1}=\pm1)
    =\frac{1}{\sqrt2}\varphi_0(\pm1)-\varphi_2(\pm1) 
    =0.
\end{aligned}
\]
Thus, even though $f\not\equiv 0$, it satisfies
\[
f(t_0)=0,\qquad \hat f(\omega_0)=0,\qquad \hat f(\omega_1)=0.
\]
In other words $\Lambda=\{0\}, M=\{-1,1\}$ is not a uniqueness pair for $\mathcal{H}_2$.
\qed
\end{example}

Continuing in the space $\mathcal{H}_2$, more broadly, if $A$ denotes the matrix of the $3\times 3$ stacked linear system for one measurement in the time and two measurements in the frequency domain (cf.~\eqref{eq:general_basis_inversion}), a straightforward computation reveals
\[
\det(A)
= C\cdot(\omega_0-\omega_1)\Bigl(t_0^2-\omega_0\omega_1-1+i\,t_0(\omega_0+\omega_1)\Bigr),
\]
for some $C=C(t_0,\omega_0,\omega_1) \neq 0$. Thus, if $\omega_0\neq\omega_1$ then
\[
\det(A)=0
\quad\Longleftrightarrow\quad
t_0^2-\omega_0\omega_1-1+i\,t_0(\omega_0+\omega_1)=0.
\]
For real $t_0,\omega_0,\omega_1$, this complex equation is equivalent to the two real equations
\[
t_0^2-\omega_0\omega_1-1=0,
\qquad
t_0(\omega_0+\omega_1)=0.
\]
The system is singular, $\det(A)=0$ iff both equations hold simultaneously. 

For \(t_0\neq 0\) we must have \(\omega_0=-\omega_1\), and hence $t_0^2+\omega_0^2=1$.
Therefore real singular configurations exist only for \(|t_0|<1\), in which case there are exactly two ordered solutions,
\begin{align*}
(\omega_0,\omega_1) &=\bigl(\sqrt{1-t_0^2},-\sqrt{1-t_0^2}\bigr) \\ 
(\omega_0,\omega_1) &=\bigl(-\sqrt{1-t_0^2},\sqrt{1-t_0^2}\bigr)
\end{align*}

For $t_0=0$, these conditions imply $\omega_0\omega_1 = -1$.
Fig.~\ref{fig:d3_m1_example} illustrates this case. 
It shows a heat map of $\log(\frac{\sigma_{\min}(A)}{\sigma_{\max}(A)})$ as the frequency sampling points \(\omega_0\) and \(\omega_1\) vary. 
Configurations with ratios $\sigma_{\min}/\sigma_{\max}$ below $1.85^*10^{-5}$ are considered numerically singular and are represented by contours.

\begin{figure}[h]
\centering
\includegraphics[width=0.38\textwidth]{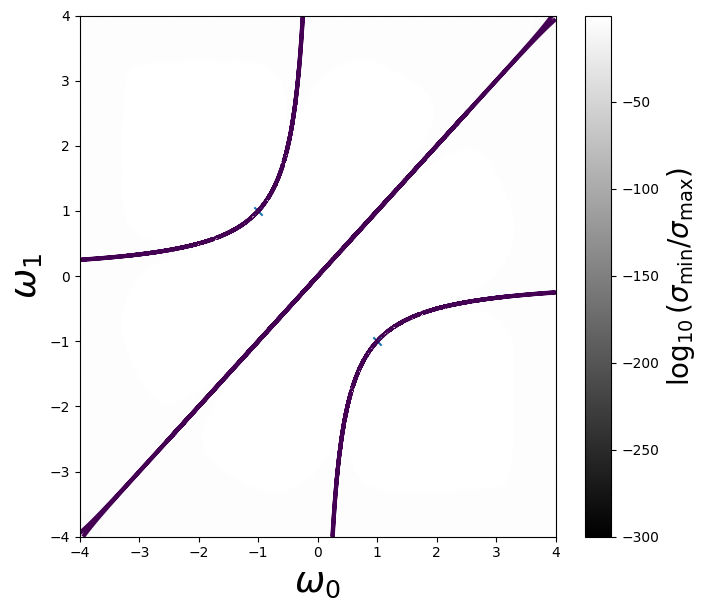}
\caption{Plot of $\log(\sigma_{\min}(A)/\sigma_{\max}(A))$ for the stacked $3\times 3$ system in $\mathcal{H}_2$, with $t_0=0$ fixed, as functions of $\omega_0$ and $\omega_1$. Numerically singular configurations (ratios below $1.85\times 10^{-5}$) occur along the diagonal $\omega_0=\omega_1$ (corresponding to repeated measurements) and the hyperbola $\omega_0\omega_1=-1$.}
\label{fig:d3_m1_example}
\end{figure}

From Fig.~\ref{fig:d3_m1_example}, it is evident that matrix inversion is possible for almost all points in the sampling space, which therefore yield unique signal reconstruction. 
As expected from our analysis, reconstruction fails along the diagonal, $\omega_0=\omega_1$, and also along two symmetric contours $\omega_0\omega_1 = -1$ passing through the counterexample points $(1,-1)$ and $(-1,1)$ from Example~\ref{examp:nonunique_f}.

More generally, the zero set of a non-zero polynomial in $m$ variables has Lebesgue measure zero in $\mathbb{R}^m$, which implies that a matrix sampled at random (according to a distribution that is absolutely continuous wrt.~Lebesgue measure) is invertible with probability one. This alone does, however, not imply that square matrices composed from entries $\Phi_j(t_i)$ and $\hat \Phi_j(\omega_i)$ are invertible for almost all finite configurations of $(t_i)_{i=0}^{K-1},(\omega_i)_{i=0}^{L-1}$.
A standard condition is that the mapping 
\[
\left ( (t_i)_{i=0}^{K-1},(\omega_i)_{i=0}^{L-1} \right ) \mapsto\det(A)
\]
is real-analytic. In this case the zero set has Lebesgue measure zero. This is true for the Hermite functions, but providing precise conditions for uniqueness pairs for $\mathcal H_N$ is a difficult task.

\subsection{Finite RKHS sampling with two-sided samples}
If $\mathcal{X}$ is a finite-dimensional RKHS, $\dim(\mathcal X)=N+1<\infty$,
one may choose a basis $\{\Phi_0,\dots,\Phi_N\}$ and identify $\mathcal X=\mathrm{span}\{\Phi_i\}$
equipped with an inner product specified by a symmetric positive definite Gram matrix
$G\in\mathbb R^{(N+1)\times (N+1)}$, i.e., for $f=\sum_i a_i\Phi_i$ and $g=\sum_i b_i\Phi_i$, set
$\langle f,g\rangle_{\mathcal X}=a^\top G\,b$. The reproducing kernel then admits the explicit
representation
\[
K(x,y)=\vec\Phi(x)^\top G^{-1}\vec\Phi(y),
\]
where $\vec\Phi(x)=(\Phi_0(x),\dots,\Phi_N(x))^\top$; equivalently, every finite-dimensional RKHS
corresponds to a feature map into $\mathbb R^{N+1}$ with a weighted Euclidean inner product, and the
reproducing property follows from the simple computation
\begin{align*}
\langle f,K_x\rangle_\mathcal X &= a^\top G ((\vec\Phi(x)^\top G^{-1})_0,...,(\vec\Phi(x)^\top G^{-1})_N)\\
&=\sum_ia_i\Phi_i(x)=f(x).
\end{align*}
One-sided sampling is governed by Eq.~\eqref{eq:RKHSExpansion}. To incorporate two-sided samples we will need an appropriate mathematical framework, which is provided in Def.~\ref{def:twosidedRKHS} and the subsequent discussion below. For now we prefer to keep the exposition heuristic - assuming the following quantities are all well-defined to convey the intuition of two-sided approach. As the kernel depends on two variables, we specify, which variable the Fourier transform applies to, following the convention:
\begin{align*}
\text{Fourier } &\text{transform applied to first variable of $K(t,\cdot)$:}\\
&{F}\big(K_s(t)\big)
=
\widehat{K_s}(\omega)\ \textnormal{(wide hat)},\\
\text{Fourier } &\text{transform applied to second variable of $K(\cdot,s)$:}\\
&{F}\big(K_s(t)\big)
=
\overline{L_\xi(t)}.
\end{align*}
Suppose now that $K$ samples can be obtained from \(f(t)\) and $L$ samples can be obtained from \(\hat f(\omega)\). Consider the two-sided expansion formula, which we will justify formally using a two-sided representer theorem, Thm.~\ref{thm:representer}, below:
\begin{align}\label{eq:toJustify}
f(t)
=
\sum_{i=0}^{K-1} \alpha_i\,K_{t_i}(t)
\;+\;\sum_{i=0}^{L-1} \beta_i\,L_{\omega_i}(t).
\end{align}
Applying ${F}$ to this equation yields
\[
\hat{{f}}(\omega)
=
\sum_{i=0}^{K-1} {\alpha}_i\,\widehat{K_{t_i}}(\omega)
\;+\;
\sum_{i=0}^{L-1}{\beta}_i\,\widehat{{L}_{\omega_i}}(\omega).
\]
Suppose we observe measurements $\vec c$ from $f(t)$ and $\vec{\hat c}$ from $\hat f(\omega)$, then the two equations yield a stacked system of equations:
\begin{align}
\begin{pmatrix}
{\vec{c}}\\[1mm]
\vec{{\hat{c}}}
\end{pmatrix}
=
\begin{pmatrix}
(K_{t_i}(t_j)) & ({L_{\omega_i}}(t_j))\\[1mm]
(\widehat{K_{t_i}}(\omega_j)) & {(\widehat{L_{\omega_i}}}(\omega_j))
\end{pmatrix}
\begin{pmatrix}
(\vec\alpha)\\
(\vec\beta)
\end{pmatrix}.\label{eq:RKHSOptimization}
\end{align}
Here we write $\left(K_{t_i}(t_j)\right)$ for the $K\times K$ matrix with entries $K_{t_i}(t_j)$, $0\leq i,j\leq K-1$, $\left(L_{\omega_i}(t_j)\right)$ for the $K\times L$ matrix with entries $L_{\omega_i}(t_j)$, $0\leq i\leq L-1, 0\leq j\leq K-1$, $(\widehat{K_{t_i}}(\omega_j))$ for the $L\times K$ matrix with entries $\widehat{K_{t_i}}(\omega_j)$, $0\leq i\leq K-1$ and $0\leq j\leq L-1$, finally, $(\widehat{L_{\omega_i}}(\omega_j))$ for the $L\times L$ matrix with entries $\widehat{L_{\omega_i}}(\omega_j)$, $0\leq i,j\leq L-1$. As outlined in the case of basis sampling, this system might have an empty, a singleton or an infinite solution set, but it can be treated uniformly by the Moore-Penrose pseudoinverse. We summarize the finite RKHS procedure in Algorithm~\ref{alg:finite_rkhs}. 

\begin{algorithm}[]
\caption{Finite RKHS reconstruction from two-sided samples}
\label{alg:finite_rkhs}
\begin{algorithmic}[1]
\REQUIRE Finite two-sided RKHS $\mathcal X$ with kernel $K$; sampling points in time  $\{t_i\}_{i=0}^{K-1}$ and frequency $\{\omega_i\}_{i=0}^{L-1}$; measurements $\vec c$, $\vec{\hat c}$. 
\STATE Set $K_x(s)=K(s,x)$ and define $L_\omega(x)=\overline{\widehat{K_x}(\omega)}$.
\STATE Assemble the RKHS block matrix $B$ in~\eqref{eq:RKHSOptimization}.
\STATE Set $b=(\vec c,\vec{\hat c})^\top$.
\STATE Compute $(\vec\alpha^\star,\vec\beta^\star)^\top=B^\dagger b$.
\ENSURE $f^\star(t)=\sum_i\alpha_i^\star K_{t_i}(t)+\sum_j\beta_j^\star L_{\omega_j}(t)$.
\end{algorithmic}
\end{algorithm}

\section{Reconstruction in the Infinite \\ Dimensional Case}\label{sec:reoncstruction_inINfiniteDimensions}
In signal sampling, natural sampling spaces like $\mathcal{PW}_T,\mathcal{H}$ have infinite dimension. This corresponds to a regime $N\to\infty$, where the number of measurements (cf.~eq.~\eqref{eq:general_basis_inversion} and \eqref{eq:RKHSOptimization}) \lq\lq{}grows simultaneously with the dimension of the space\rq\rq{}.
\subsection{Basis sampling with two-sided samples}
As the size of the system of equations (cf.~\eqref{eq:general_basis_inversion}, \eqref{eq:RKHSOptimization}) increases, the smallest singular value may decay toward zero, degrading conditioning and potentially leading to loss of injectivity in the limit. Thus, even if each finite-$N$ problem admits a unique solution, the large-$N$ system may become non-unique.
The notion of a uniqueness pair ensures stability in this respect.
\begin{proposition}\label{prop:uniqueness_pair}
If $(\Lambda,M)$ is a uniqueness pair for a (potentially infinite-dimensional) function space $\mathcal{X}$ and the countable system of constraints admits a solution, then the solution is unique. 
\end{proposition}
\begin{proof}
Assume $\vec\alpha^{(1)}, \vec\alpha^{(2)}$ are two solutions, then $\vec\alpha^{(1)}-\vec\alpha^{(2)}$ lies in the kernel of the linear system. The function $f$ represented by the coefficients $\{\alpha_n^{(1)}-\alpha_n^{(2)}\}$ thus satisfies $f|_\Lambda=0$ and $\hat{f}|_M=0$, which is a contradiction unless $f=0$.
\end{proof}

The Paley-Wiener model treats signals as exactly bandlimited, so that the signals' complexity is determined by the bandwidth. A more flexible model, developed by Landau, Pollak, and Slepian, considers functions that are simultaneously well concentrated in frequency on \([-W,W]\) and in time on \([-T,T]\). The associated time-frequency concentration operator quantifies how compatible these two localization requirements are. The spectrum of this operator has a sharp plunge~\cite{Slepian1961,Slepian1961b,Landau1962}: approximately ${2TW}/{\pi}$ eigenvalues are close to \(1\), while the rest are close to \(0\), up to a narrow transition region. Thus, the space of signals that can be both time- and band-concentrated behaves, to first order, like a finite-dimensional space of dimension \(2TW/\pi\). Equivalently, this is the number of orthogonal modes that carry essentially all of the energy of such signals.

This motivates a finite-dimensional analysis of two-sided reconstruction in Sec.~\ref{sec:reconstruction_inFiniteDimensions}. In particular, any reconstruction based on $K$ time samples and $L$ frequency samples must satisfy the necessary condition $K+L\geq {2TW}/{\pi}$. The prolate spheroidal wave functions, being the eigenfunctions of the time-frequency concentration operator, provide the natural basis for the class of signals that are simultaneously concentrated in time and frequency. When the dimension becomes large, these functions approximate Hermite functions and the respective model spaces \lq\lq{}converge\rq\rq{} to the Fourier-symmetric Sobolev space $\mathcal{H}$.

\subsubsection{The Fourier-symmetric Sobolev space} Using integration by parts and Plancherel's identity, the norm on $\mathcal H$ can be written in the form
$$||f||^2_{\mathcal H}=2\int_{\mathbb R} f \overline{H f} \textnormal{d}x,\ \textnormal{where}\ H=\frac{1}{2}\left(x^2-\frac{\partial^2}{\partial x^2}\right).$$
Thus the natural scalar product for $\mathcal{H}$ is $\langle f,g\rangle = 2\int_{\mathbb{R}} f \overline{H g}\:\textnormal{d}x$. $H$ is the classical Hamiltonian operator of the quantum harmonic oscillator. It is well-known to be a self-adjoint and positive operator, whose eigenvectors are the Hermite functions
$$H\varphi_n=E_n\varphi_n,\ \ E_n=n+\frac{1}{2}.$$
Thus in the Hermite basis $\mathcal H$ is the weighted $l^2$-space with norm
$$||f||_\mathcal H=\sum_{n=0}^\infty E_n|\alpha|^2,\quad f=\sum_{n=0}^\infty \alpha_n\varphi_n.$$

The Fourier-invariance of $\mathcal H$ manifests itself in the operator identity $H=\hat H$ in the sense that
\[
H=\frac12\left(x^2-\frac{\partial^2}{\partial x^2}\right),
\qquad
\hat H=\frac12\left(k^2-\frac{\partial^2}{\partial k^2}\right)
\]
and the fact that Hermite functions are eigenfunctions of the Fourier transform. $\mathcal H$ constitutes the natural infinite-dimensional extension of the finite Hermite function spaces considered in Sec.~\ref{sec:finite_basis_sampling}. Uniqueness pairs for reconstruction in two-sided sampling in $\mathcal H$ have been identified recently in terms of a density condition $(\Lambda,M)$. The pair $(\Lambda = \{\lambda_j\}_j,M=\{\mu_j\}_j)$ is, respectively, called supercritical $(<)$ or subcritical $(>)$ if simultaneously the following conditions hold\footnote{Our threshold differs from~\cite{kulikov2023fourier} since we use the unitary convention for the Fourier transforms.}
\begin{align*}
&\lim\sup\{|\lambda_j|(\lambda_{j+1}-\lambda_j)\}\lessgtr\pi,\\ &\lim\sup\{|\mu_j|(\mu_{j+1}-\mu_j)\}\lessgtr\pi.
\end{align*}
\begin{proposition}[Uniqueness pairs in $\mathcal{H}$~\cite{kulikov2023fourier}]\label{prop:uniqueness_pair_inH}
$\ $
\begin{enumerate}
    \item If $(\Lambda,M)$ is supercritical then it is a uniqueness pair for $\mathcal H$.
    \item If $(\Lambda,M)$ is subcritical then it is a non-uniqueness pair for $\mathcal H$.
\end{enumerate}
\end{proposition}
Algorithm~\ref{alg:infinite_basis} summarizes the discussion above on the effective finite-dimensional use of basis expansions and the exact uniqueness statements for $\mathcal H$.
\begin{algorithm}[h]
\caption{Two-sided basis sampling in infinite-dimensional spaces}
\label{alg:infinite_basis}
\begin{algorithmic}[1]
\REQUIRE Two-sided samples on sets $\Lambda=\{t_i\}$ and $M=\{\omega_j\}$; an infinite-dimensional basis model for $\mathcal X$. \vspace{0.3em}
\item[\textbullet]  \textbf{Effective finite-dimensional approximation.}
\STATE If signals are essentially concentrated in time on $[-T,T]$ and in frequency on $[-W,W]$, set the effective dimension $D_{\rm eff}\approx 2TW/\pi$.
\STATE Choose a finite reconstruction space of dimension fit into this concentration model and require a sample budget at least on this scale, $K+L \approx D_{\rm eff}$.
\STATE Solve the finite stacked system as in the finite case. \vspace{0.3em}
\item[\textbullet] \textbf{Exact reconstruction under Fourier-symmetric Sobolev space.}
\STATE If $\mathcal X=\mathcal H$, use the Hermite expansions in both time and Fourier domains: 
\begin{align}
    f&=\sum_{n=0}^{\infty}\alpha_n\varphi_n \notag \\
    \widehat{f}&=\sum_{n=0}^{\infty}\alpha_n (-i)^n\varphi_n \notag 
\end{align} 
\STATE Check if $(\Lambda,M)$ is sufficiently dense according to Prop.~\ref{prop:uniqueness_pair_inH} for uniqueness
\ENSURE the finite approximation in the first case, or the exact uniqueness conclusion in the second case
\end{algorithmic}
\end{algorithm}

\subsection{RKHS sampling with two-sided samples}
We begin providing the formal mathematical framework for two-sided RKHS sampling.
\begin{definition}[Two-sided RKHS]\label{def:twosidedRKHS} Let $\mathcal X$ be a Hilbert space of functions $f:X\rightarrow\mathbb C$, where $X$ is a domain equipped with the Fourier transform. We call $\mathcal{X}$ a two-sided reproducing kernel Hilbert space if, writing $X'$ for the Fourier dual domain,
for every $t\in X$ the point evaluation functional in the time domain
$$l_t:\mathcal{X}\rightarrow\mathbb C,\ f\mapsto f(t)$$
and for every $\omega\in X'$ the point evaluation functional in the Fourier domain
$$\hat l_\omega:\mathcal{X}\rightarrow\mathbb C,\ f\mapsto \hat f(\omega)$$
are linear and continuous.
\end{definition}
We will focus on $X=\mathbb R$, but the definition is valid for other domains, e.g.~$X=\mathbb R^d$. The Fourier functional is of the form $\hat l_\omega(f)=\hat{f}(\omega)=F[f](\omega)$.
By the Riesz representation theorem, the functionals $l_t$, $\hat{l}_\omega$ admit representations of the form $l_t(f)=\langle f, K_t\rangle_{\mathcal X}$ and $\hat l_\omega(f)=\langle f, L_\omega\rangle_{\mathcal X}$ for certain representers $K_t,L_\omega\in\mathcal X$. Choosing $f=K_t$ in the last formula thus yields the pointwise identity
$$\widehat{K_t}(\omega)=\langle K_t, L_\omega\rangle_{\mathcal X}=\overline{L_\omega(t)}.$$
In case that $X=\mathbb R$ and $K_t$ is Fourier integrable then under the unitary convention
$${L_\omega(t)}=\overline{\widehat{K_t}(\omega)}=\frac{1}{\sqrt{2\pi}}\int_\mathbb{R}K(t,s)e^{i\omega s}\textnormal{d} s.$$
To formally justify the reduction of two-sided sampling to the finite system of equations~\eqref{eq:RKHSOptimization}, we need a representer theorem, see~\cite{Wahba1999,Christmann2008}.
\begin{theorem}[Two-sided representer theorem]\label{thm:representer}
Let $(\mathcal X,\langle\cdot,\cdot\rangle_\mathcal X)$ be a two-sided RKHS. Fix data points $t_0,...,t_{K-1}$ and $\omega_0,...,\omega_{L-1}$. Let $\Psi:\mathbb{C}^{K+L}\rightarrow\mathbb{R}\cup\{\infty\}$ be any function and let $\Omega:[0,\infty)\rightarrow\mathbb{R}\cup\{\infty\}$ be strictly increasing. Consider the optimization problem
\begin{align*}
\min_{f\in\mathcal X}\Bigg\{\Psi\Big(f(t_0),...,f(t_{K-1}),\hat{f}(\omega_0),...,\hat{f}(\omega_{L-1})\Big)&\\
+&\Omega(||{f}||_\mathcal{X})\Bigg\}.\end{align*}
If a minimizer $f^*\in\mathcal X$ exists, then it is of the form
$$f^*(t)=\sum_{i=0}^{K-1}\alpha_iK_{t_i}(t)+\sum_{i=0}^{L-1}\beta_i L_{\omega_i}(t),$$
where $K_{t}$ are Riesz representers of time domain point evaluation functionals and $L_{\omega}$ are representers of frequency domain point evaluation functionals. Moreover, $L_\omega(t)=\overline{\widehat{K_t}(\omega)}$. In the Euclidean case $X=\mathbb R$, if $K_t$ is Fourier integrable, then
$$L_\omega(t)=\frac{1}{\sqrt{2\pi}}\int_\mathbb{R}K(t,s)e^{i\omega s}\textnormal{d} s.$$
\end{theorem}
This justifies the reduction to the finite system of equations~\eqref{eq:toJustify}.
\begin{proof}
    The proof is a consequence of a wider principle that representer theorems hold for continuous linear functionals~\cite{NIPS2012_eb160de1} that even extends to Banach spaces~\cite{Zhang2009RKBS,Szehr2020}. We use the formulation of~\cite{NIPS2012_eb160de1}, which asserts that if a minimizer exists for the problem
    \begin{align*}
\min\Bigg\{\Psi\Big(L_0[f],...,L_{N-1}[f])\Big)
+\Omega(||{f}||_\mathcal{X})\Bigg\},\end{align*}
where $L_0,...,L_{N-1}$ are continuous linear functionals, then the minimizer is of the form $f^*=\sum_{i=0}^{N-1}\alpha_iR_i$ with $R_i$ the representer of $L_i$. By assumption both time and Fourier domain point evaluation functionals are continuous, and we computed the explicit form of $L_\omega(t)$ above.
\end{proof}
\begin{remark}\label{rem:countable}
    The statement remains valid for a countable family of functionals. In this case 
    $f^*\in\overline{\textnormal{span}\{R_n,\:n\in\mathbb N\}}$ and for suitable sequences of coefficients
    $$ \sum_{i=0}^{K}\alpha_i^{(K)}K_{t_i} + \sum_{i=0}^{L}\beta_i^{(L)} L_{\omega_i}\rightarrow f^*$$
    as $K,L\rightarrow\infty$.
\end{remark}

We illustrate this representer theorem with the familiar example $\mathcal H$.
\subsubsection{The Fourier-symmetric Sobolev space as a two-sided RKHS}
The Fourier functionals on $\mathcal H$ are continuous. We have that
\begin{align*}|\hat{f}(\omega)|&\leq\frac{1}{\sqrt{2\pi}}\int_\mathbb R|f(t)|\sqrt{1+t^2}\frac{\textnormal{d}t}{\sqrt{1+t^2}}\\
&\leq\left(\frac{\pi}{2}\int_\mathbb{R}|f(t)|^2(1+t^2)\textnormal{d}t\right)^{1/2},
\end{align*}
where the first inequality uses the definition of the Fourier transform and second inequality is Cauchy-Schwarz. The right-hand side is bounded since $f\in\mathcal H$. Making use of the inverse Fourier transform, the same argument shows that the point evaluation functionals are continuous.

The underlying vector space $\mathcal H$ admits several equivalent Hilbert norms that yield RKHS structures. A canonical choice that extends the finite dimensional discussion and preserves the symmetry with respect to the Fourier transform is to interpret the positive operator $H$ as an infinite Gram matrix: Consider the scalar product defined for the Hermite basis expansions $f=\sum_{n=0}^\infty \alpha_n\varphi_n$, $g=\sum_{n=0}^\infty \beta_n\varphi_n$ by
$$\langle f,g\rangle_{\mathcal H}=\sum_{n=0}^\infty {\alpha}_n\bar{\beta}_n E_n.$$
The symmetric and positive function 
$$K(x,y)=\sum_{n=0}^\infty E_n^{-1}\varphi_n(x)\varphi_n(y)$$
induces a reproducing kernel.
Using the completeness of Hermite functions and the eigenvalue equation of $H$, direct computation reveals that
\begin{align*}
\langle f,K_y\rangle_{\mathcal H} &= \sum_{n=0}^\infty {\alpha}_nE_nE_n^{-1}\varphi_n(y)=f(y).
\end{align*}
This yields an explicit RKHS reconstruction scheme for two-sided sampling. The presented kernel representation together with the representer theorem~\ref{thm:representer} yield the linear system~\eqref{eq:RKHSOptimization}, which can be solved by standard tools such as the pseudoinverse. In the finite-dimensional case the kernel can be truncated to a finite sum.

There are two paths to ensure uniqueness of reconstruction:

First, uniqueness is guaranteed if $\Psi$ is convex and the regularizer $\Omega(||f||_\mathcal X)$ is strictly convex in $f$. In this case the entire optimization functional is strictly convex and the minimizer is unique if it exists~\cite{NIPS2012_eb160de1}. The most important example is Tikhonov regularization, where the functional $\Psi+\lambda||f||^2_\mathcal X$ is strictly convex if $\Psi$ is convex.

Second, there is the injectivity of the sampling operator itself. This is ensured if the two-sided samples form a uniqueness pair. In this situation unique reconstruction holds even in the absence of a regularizer. Naturally, finding precise conditions is much harder in this situation. In this situation the formula of Rem.~\ref{rem:countable} takes the form of a unique representer. This leads to reconstruction formulas structurally analogous to those of Radchenko-Viazovska~\cite{radchenko2019fourier}. Notice, however, that their reconstruction nodes $\Lambda=M=\{\sqrt{2\pi n}\}_{n\in\mathbb{Z}_{\geq0}}$ are located at the edge of the regions of Prop.~\ref{prop:uniqueness_pair_inH} since
\[
\lim_{n\rightarrow\infty}\sqrt{2\pi n}(\sqrt{2\pi (n+1)}-\sqrt{2\pi n})=\pi.
\]
No immediate conclusion about these nodes is possible here. Naturally, this path to the study of unique reconstruction is much harder. Algorithm~\ref{alg:infinite_rkhs} summarizes the RKHS discussion above.

\begin{algorithm}[h]
\caption{Two-sided RKHS sampling in infinite-dimensional spaces}
\label{alg:infinite_rkhs}
\begin{algorithmic}[1]
\REQUIRE Two-sided RKHS $\mathcal X$ with kernel $K$; samples on sets $\Lambda=\{t_i\}$ and $M=\{\omega_j\}$.
\STATE Set $K_x(s)=K(s,x)$ and define $L_\omega(x)=\overline{\widehat{K_x}(\omega)}$.
\STATE Choose finite subsets $\Lambda_N\subset\Lambda$ and $M_N\subset M$.
\STATE Assemble the finite RKHS block matrix as in~\eqref{eq:RKHSOptimization} and solve in order to obtain
\begin{equation}
    f_N^\star \in \operatorname{span}\{K_t:t\in\Lambda_N\} + \operatorname{span}\{L_\omega:\omega\in M_N\} \notag
\end{equation}\vspace{-1em}
\STATE Increase $\Lambda_N,M_N$ to refine the approximation.
\STATE By Rem.~\ref{rem:countable}, the countable sample solution, when it exists, lies in
$\overline{\operatorname{span}\{K_t,L_\omega:t\in\Lambda,\omega\in M\}}$.
\STATE If $f_N^\star\to f^\star$ in $\mathcal X$, return the limiting representer reconstruction $f^\star$.
\STATE For uniqueness, use injectivity of the two-sided sampling operator; for $\mathcal X=\mathcal H$, Prop.~\ref{prop:uniqueness_pair_inH} gives a sufficient density condition. Strict convex regularization gives uniqueness of the regularized solution.
\ENSURE Finite approximants $f_N^\star$ and, when convergent, a limiting reconstruction $f^\star$; uniqueness conclusion
\end{algorithmic}
\end{algorithm}

\section{Numerical Experiments}\label{sec:numerical_exps}

\subsection{Comparison under a fixed sample budget}

In this section, we consider $N$-dimensional Hermite function space $$\mathcal{H}_{N-1}=\mathrm{span}\{\varphi_0,\dots,\varphi_{N-1}\},$$
and compare one- and two-sided sampling under the same total sample budget
\[
D = K+L = N.
\]
We measure reconstruction stability by the condition number $\sigma_{\max}(A)/\sigma_{\min}(A)$ of the matrix $A$.
Smaller values indicate a better conditioned inversion problem and lower sensitivity to perturbations. 
\begin{figure}[h]
\centering
\includegraphics[width=0.48\textwidth]{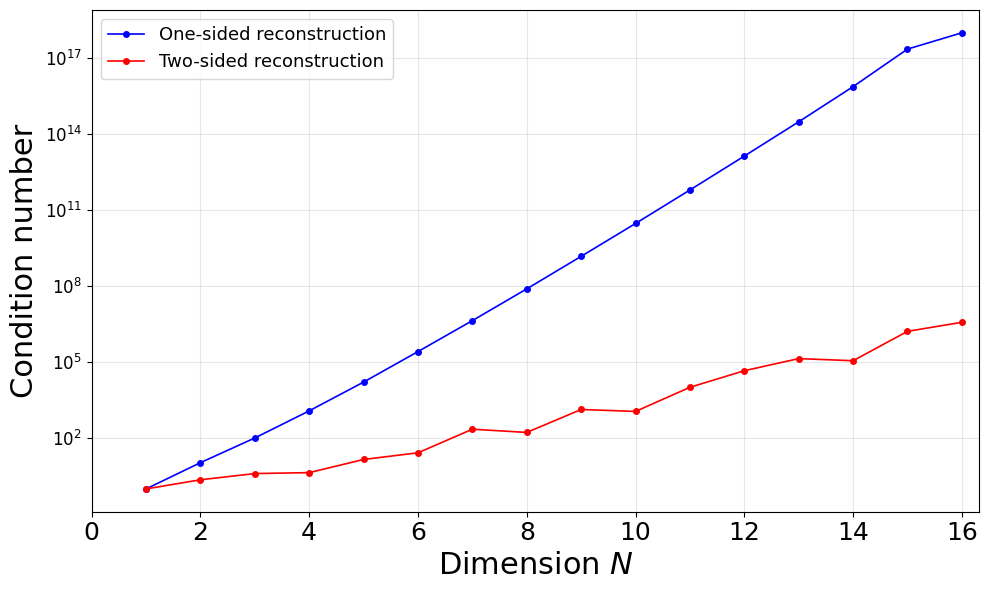}
\caption{Condition numbers for one-sided and two-sided sampling with fixed total budget $D=N$. Time samples are taken on $[1,2]$ and frequency samples on $[-1,0]$. In both cases, the samples are equispaced.}
\label{fig:plot1}
\end{figure}

In our first experiment, for  pure time-domain sampling, we take $K=D$ time samples and no frequency samples ($L=0$). For each $D$, the sampling points are chosen on an equispaced grid over the interval $[1.0,2.0]$ with spacing of $1/D$. For two-sided sampling, we retain the same total budget of $D$ samples but divide it between the time and frequency samples by
\[
K=\left\lceil \frac{D}{2}\right\rceil,\qquad L=D-K=\left\lfloor \frac{D}{2}\right\rfloor.
\]
The sampling points are then chosen on equispaced grids in time over $[1.0,2.0]$ and in frequency over $[-1.0,0]$.

As we can see in Fig.~\ref{fig:plot1}, the condition number increases with the dimension $N$, which indicates that the reconstruction problem becomes progressively more challenging in both cases even though the number of available samples $D$ also increases ($D=N$).
Nevertheless, two-sided reconstruction yields typically smaller condition numbers than one-sided sampling and is therefore more stable numerically.

\begin{figure}[h]
\centering
\includegraphics[width=0.48\textwidth]{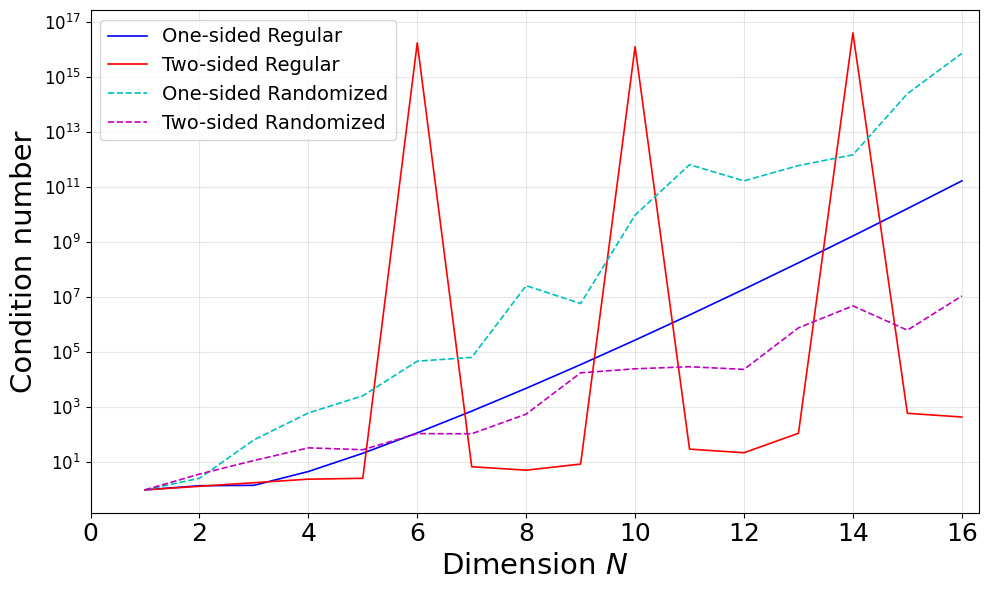}
\caption{Condition numbers when both time and frequency samples lie in $[-1,1]$. Equispaced two-sided sampling is singular for some $N$, whereas random sampling avoids these failures and remains better conditioned.}
\label{fig:plot2}
\end{figure}

This behavior does not persist for all configurations. 
In Fig.~\ref{fig:plot2}, where both the time and frequency samples are taken on equispaced grids in $[-1,1]$, the resulting mixed sampling configuration becomes non-invertible for some values of $N=D$. This is consistent with the uniqueness-pair discussion above: certain aligned choices of time and frequency nodes can fail to determine the signal uniquely, leading to a singular sampling matrix. By contrast, in the same figure, we can see that if the sampling points are chosen randomly (uniformly) rather than equispaced within the same $[-1,1]$ interval, two-sided sampling once again outperforms one-sided sampling consistently. 
This suggests that some degree of misalignment between the time and frequency sampling sets is useful. 
In particular, when the relevant uniqueness sets are not known a priori, randomization may serve as an effective practical strategy. Related forms of misalignment are also known to play a fundamental role in compressed sensing \cite{candes2006robust}.

\begin{figure}[h]
\centering
\includegraphics[width=0.48\textwidth]{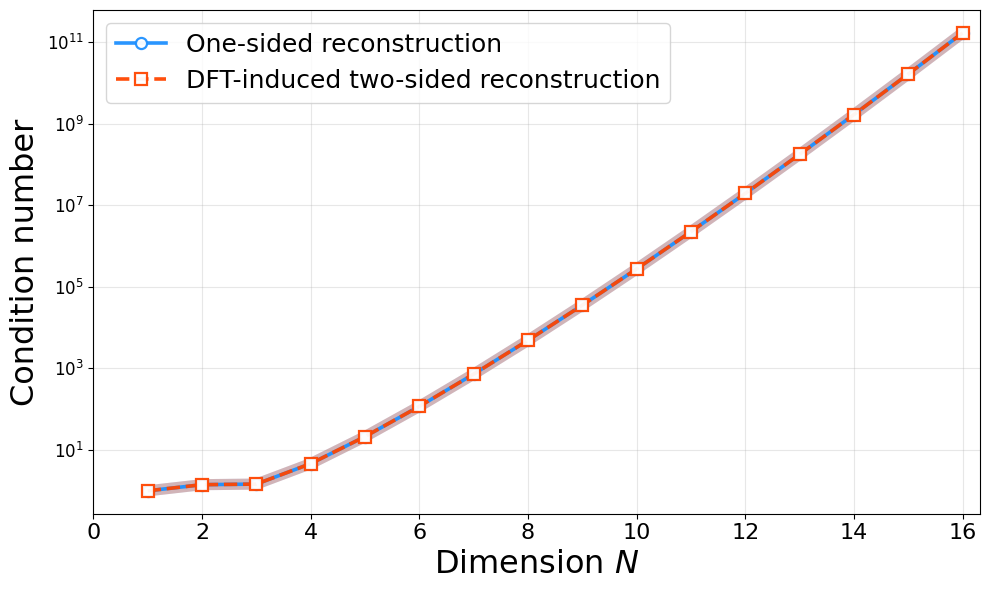}
\caption{Condition numbers are unchanged when part of the time-domain data is post-processed by the DFT and treated as frequency-domain data.}
\label{fig:plot4}
\end{figure}

Finally, we consider a different experiment. We begin with $D=N$ time-domain samples and apply the discrete Fourier transform to a subset of the measurement vector, treating the resulting coefficients as surrogate frequency-domain data. As expected, this leaves the condition number unchanged, cf.~Fig.~\ref{fig:plot4}: invertible post-processing of existing measurements cannot increase the information content.

\subsection{Bandlimited signal space}

We now turn to the standard sinc interpolation setting. As in the previous section, we investigate the condition number and compare the classical one-sided time-domain sampling with two-sided sampling under a fixed sampling budget $D=N$. 
More precisely, consider the finite-dimensional space
\[
 \mathcal{PW}_{\pi}^{N-1}=\mathrm{span}\{s_0,\dots,s_{N-1}\},
\]
where
\[
s_n(t)=\mathrm{sinc}(t-n),\qquad n=0,1,\dots,N-1,
\]
with \(T=1\). Under the unitary Fourier convention, we have
\[
\widehat{\mathrm{sinc}}(\omega)= \frac{1}{\sqrt{2\pi}}\mathbf{1}_{|\omega|\le \pi},
\qquad
\widehat{s_n}(\omega)= \frac{1}{\sqrt{2\pi}} e^{-i\omega n}\,\mathbf{1}_{|\omega|\le \pi}.
\]
In other words, \(\mathcal{PW}_{\pi}^{N-1}\) is bandlimited in the frequency domain and integer shifts in time translate into phase shifts in frequency. 

For frequency-domain sampling, we consider the interval \([-3,3]\), since sampling outside the range $[-\pi,\pi]$ carries no additional information.
In the time-domain, as opposed to the previous section, we consider sampling intervals that increase with increasing sample size \(D\), i.e., $[\left\lceil -\frac{D}{2} + 1\right\rceil, \cdots, \left\lfloor \frac{D}{2}\right\rfloor]$.
In both settings, the sampling locations are chosen uniformly at random from the corresponding intervals, as in Fig~\ref{fig:plot2}.

\begin{figure}[h]
\centering
\includegraphics[width=0.48\textwidth]{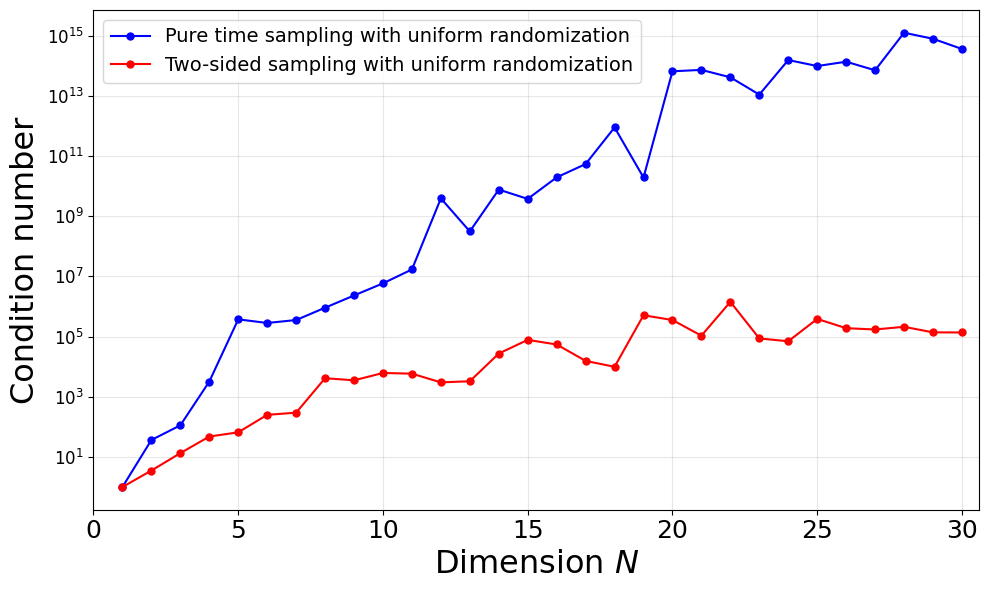}
\caption{Condition numbers in a finite-dimensional bandlimited space spanned by integer-shifted sinc functions. The time-sampling interval increases with the dimension of the function space, while the frequency-sampling interval is fixed at $[-3,3]$. Sampling points are chosen uniformly at random.}
\label{fig:sinc_example}
\end{figure}

Fig.~\ref{fig:sinc_example} shows that, as in the Hermite-generated case, the reconstruction problem becomes increasingly ill-conditioned as \(N\) grows. Moreover, the two-sided sampling scheme consistently yields smaller condition numbers than the corresponding one-sided scheme in the bandlimited setting as well.

\section{Spectrum Monitoring Application}\label{sec:spectrum_monitoring}

In spectrum monitoring, the goal is to observe and analyze the radio-frequency environment in order to detect signals of interest, identify interference, and ensure efficient spectrum usage \cite{axell2012spectrum,spectrumLunden2015}.
This is critical for applications such as telecommunication networks and defence systems.
In such settings, the receiver digitizes a time-domain signal and then applies a discrete Fourier transform to convert those samples into spectra.
However, acquiring and storing Nyquist-rate time samples over very large bandwidths is often prohibitively expensive.
In practice, this leads to memory bottlenecks due to insufficient storage.
Nevertheless, reconstructing the signal afterward is still required for many tasks.
For instance, once an interferer or anomaly is noticed in the spectrum, the underlying waveform is needed to characterize the signal.

\begin{figure}[h]
\includegraphics[width=0.48\textwidth]{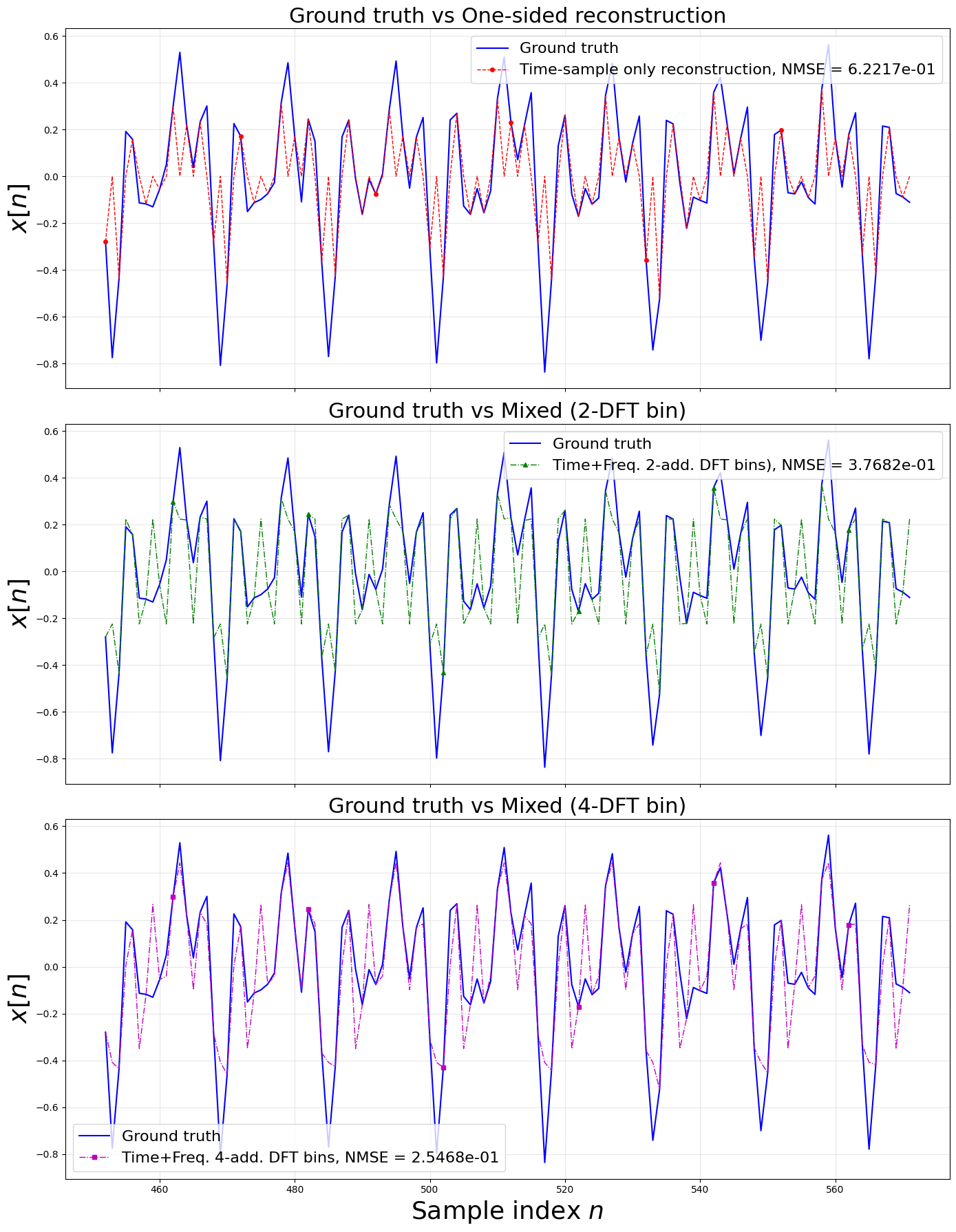}
\caption{Time-domain signal reconstruction in a spectrum-monitoring system: ground truth, reconstruction from time-domain samples only, reconstruction from time samples plus $2$ DFT bins, and reconstruction from time samples plus $4$ DFT bins.}
\label{fig:spectrum_monitoring}
\end{figure}

The signal to analyze is usually not available at full rate, because the raw time samples were too costly to retain. 
The receiver, however, has already produced the spectra, so some of the largest magnitude-bins can be retained at little cost.
Here, depending on the propagation environment, the observed radio channels may exhibit some sparsity in the frequency domain so that a small number of bins can carry substantial information. 
Recall, however, that our approach does not require sparsity in general.
These bins can serve a dual purpose by supporting both the operator's inspection and the subsequent reconstruction of the waveform. 
Indeed, as we demonstrate next, the reconstruction can be improved if the reduced time samples are supplemented with the frequency samples in the spectra in a joint reconstruction.

Namely, during acquisition, the system samples the signal over time and computes a discrete Fourier transform of length $Z$. 
For each window of $Z$ samples, it is stored every other time sample together with a few frequency bins of largest magnitude. 
The full spectrum of length $Z$ is computed as the data arrive, while only the largest complex bins are retained. 
At a later time, whether prompted by a trigger event or by an operator who wishes to recover a segment, the signal is rebuilt from the data stored as stated. 
We examine whether a block of $Z = 1024$ samples can be recovered when due to storage limits there is only access to $Z/2$ time samples from the block, taken as every other sample, together with the $2$ or $4$ strongest frequency bins.
The signal is composed of four tones and is corrupted by additive white Gaussian noise at an \(\texttt{SNR} = 16\). 

In Fig.~\ref{fig:spectrum_monitoring}, we compare three cases for the signal reconstruction. 
The first case uses only the stored $Z/2$ raw time samples, the second incorporates the $2$ bins of largest magnitude, and the third incorporates the $4$ bins of largest magnitude. 
Each reconstruction is performed in the least squares sense. 
It can be seen that incorporating the strongest complex frequency bins improves the reconstruction over the time samples alone, and $4$ bins improve it further over $2$. Fig.~\ref{fig:spectrum_monitoring} also reports the normalized MSE averaged over ten independent experiments. 
The error falls from $0.62$ with time samples alone, to $0.37$ with $2$ frequency bins, and to $0.25$ with $4$.

\section{Conclusion}

It is common in signal processing to treat time and frequency-domain information as if they are independent dimensions, with their relationship typically entering the discussion only through uncertainty principles. 
This article presents a complementary perspective by examining joint reconstruction from two-sided samples in both domains.
We analyzed this problem in finite-dimensional model spaces and related it to infinite-dimensional settings through RKHS and uniqueness results. 
These connections show that two-sided sampling is not only a computational paradigm in finite dimensions but also part of a broader functional-analytic framework.
While our article focused on several representative applications, we believe that the main message is broader. 
We expect that this two-sided viewpoint will be useful in other reconstruction problems where measurements in dual domains are naturally available and should be exploited jointly.


\section*{Acknowledgment}
The authors are grateful to Andrii Bondarenko, Visa Koivunen, and Kristian Seip for valuable discussions. This work was supported by the Swiss National Science Foundation, SNF grant No. CRSK-2\_229036.

\bibliographystyle{IEEEtran}
\bibliography{ref.bib}

@article{liang1994efficient,
  title={An efficient method for dynamic magnetic resonance imaging},
  author={Liang, Zhi-Pei and Lauterbur, Paul C},
  journal={IEEE Transactions on Medical Imaging},
  volume={13},
  number={4},
  pages={677--686},
  year={1994}
}

@ARTICLE{zhi1996constrained,
  author={Zhi-Pei Liang and Lauterbur, P.C.},
  journal={IEEE Engineering in Medicine and Biology Magazine}, 
  title={Constrained imaging: overcoming the limitations of the {Fourier} series}, 
  year={1996},
  volume={15},
  number={5},
  pages={126-132},
  doi={10.1109/51.537069}}

@article{weizman2016reference,
  title={Reference-based {MRI}},
  author={Weizman, Lior and Eldar, Yonina C and Ben Bashat, Dafna},
  journal={Medical Physics},
  volume={43},
  number={10},
  pages={5357--5369},
  year={2016}
}

@article{matsuki2009spectroscopy,
  title={Spectroscopy by integration of frequency and time domain information for fast acquisition of high-resolution dark spectra},
  author={Matsuki, Yoh and Eddy, Matthew T and Herzfeld, Judith},
  journal={Journal of the American Chemical Society},
  volume={131},
  number={13},
  pages={4648--4656},
  year={2009}
}

@article{ferreira2002interpolation,
  title={Interpolation and the discrete {Papoulis-Gerchberg} algorithm},
  author={Ferreira, Paulo Jorge SG},
  journal={IEEE Transactions on Signal Processing},
  volume={42},
  number={10},
  pages={2596--2606},
  year={2002},
  publisher={IEEE}
}

@inproceedings{hirani1996combining,
  title={Combining frequency and spatial domain information for fast interactive image noise removal},
  author={Hirani, Anil N and Totsuka, Takashi},
  booktitle={Proc. Annual Conference on Computer Graphics and Interactive Techniques},
  pages={269--276},
  year={1996}
}

@techreport{3gpp_ts38211,
  author      = "{3GPP}",
  title       = "{NR; Physical channels and modulation}",
  institution = "{3GPP}",
  number      = "{TS 38.211}",
  year        = "{2024}",
}

@article{gerchberg1974super,
  title={Super-resolution through error energy reduction},
  author={Gerchberg, RW},
  journal={Optica Acta: International Journal of Optics},
  volume={21},
  number={9},
  pages={709--720},
  year={1974}
}

@inproceedings{marques2006papoulis,
  title={The {Papoulis-Gerchberg} algorithm with unknown signal bandwidth},
  author={Marques, Manuel and Neves, Alexandre and Marques, Jorge S and Sanches, Joao},
  booktitle={International Conference Image Analysis and Recognition},
  pages={436--445},
  year={2006}
}

@ARTICLE{papoulis1975,
  author={Papoulis, A.},
  journal={IEEE Transactions on Circuits and Systems}, 
  title={A new algorithm in spectral analysis and band-limited extrapolation}, 
  year={1975},
  volume={22},
  number={9},
  pages={735-742},
  doi={10.1109/TCS.1975.1084118}}

@ARTICLE{axell2012spectrum,
  author={Axell, Erik and Leus, Geert and Larsson, Erik G. and Poor, H. Vincent},
  journal={IEEE Signal Processing Magazine}, 
  title={Spectrum Sensing for Cognitive Radio : State-of-the-Art and Recent Advances}, 
  year={2012},
  volume={29},
  number={3},
  pages={101-116},
  doi={10.1109/MSP.2012.2183771}}

@techreport{itur_sm2117_2018,
  author      = "{{ITU-R}}",
  title       = "{Data format definition for exchanging stored {I/Q} data for the purpose of spectrum monitoring}",
  number      = "{SM.2117-0}",
  address     = "{Geneva, Switzerland}",
  month       = "{Sep.}",
  year        = {2018},
  institution = "{International Telecommunication Union}",
  url         = {https://www.itu.int/rec/R-REC-SM.2117}
}

@ARTICLE{spectrumLunden2015,
  author={Lunden, Jarmo and Koivunen, Visa and Poor, H. Vincent},
  journal={IEEE Signal Processing Magazine}, 
  title={Spectrum Exploration and Exploitation for Cognitive Radio: Recent Advances}, 
  year={2015},
  volume={32},
  number={3},
  pages={123-140},
  doi={10.1109/MSP.2014.2338894}}

@book{bahai2004multi,
  title={Multi-carrier digital communications: {Theory} and applications of {OFDM}},
  author={Bahai, Ahmad RS and Saltzberg, Burton R and Ergen, Mustafa},
  year={2004},
  publisher={Springer}
}

@book{rappaport2010wireless,
  title={Wireless communications: Principles and practice},
  author={Rappaport, Theodore S},
  year={2010},
  publisher={Pearson Education}
}

@ARTICLE{landau1967,
  author={Landau, H.J.},
  journal={Proceedings of the IEEE}, 
  title={Sampling, data transmission, and the {Nyquist} rate}, 
  year={1967},
  volume={55},
  number={10},
  pages={1701-1706},
  doi={10.1109/PROC.1967.5962}}

@article{whittaker1915xviii,
  title={On the functions which are represented by the expansions of the interpolation-theory},
  author={Whittaker, Edmund Taylor},
  journal={Proceedings of the Royal Society of Edinburgh},
  volume={35},
  pages={181--194},
  year={1915},
  publisher={Royal Society of Edinburgh Scotland Foundation}
}

@article{candes2006robust,
  title={Robust uncertainty principles: Exact signal reconstruction from highly incomplete frequency information},
  author={Cand{\`e}s, Emmanuel J and Romberg, Justin and Tao, Terence},
  journal={IEEE Transactions on Information Theory},
  volume={52},
  number={2},
  pages={489--509},
  year={2006},
  publisher={IEEE}
}

@ARTICLE{unser2000,
  author={Unser, Michael},
  journal={Proceedings of the IEEE}, 
  title={Sampling-50 years after {Shannon}}, 
  year={2000},
  volume={88},
  number={4},
  pages={569-587},
  doi={10.1109/5.843002}}

@article{radchenko2019fourier,
  title={Fourier interpolation on the real line},
  author={Radchenko, Danylo and Viazovska, Maryna},
  journal={Publications math{\'e}matiques de l'IH{\'E}S},
  volume={129},
  number={1},
  pages={51--81},
  year={2019},
  publisher={Springer}
}

@article{zelent2025time,
  title={Time frequency localization in the {Fourier} Symmetric {Sobolev} space},
  author={Zelent, Denis},
  journal={arXiv:2505.04286},
  year={2025}
}

@article{szehr2025spectral,
  title={Spectral Criteria for Unique Signal Recovery from Two-Sided Sampling},
  author={Szehr, Oleg},
  journal={arXiv:2509.14953},
  year={2025}
}

@book{seip2004interpolation,
  title={Interpolation and sampling in spaces of analytic functions},
  author={Seip, Kristian},
  volume={33},
  year={2004},
  publisher={American Mathematical Society}
}

@article{kulikov2023fourier,
  title={Fourier uniqueness and non-uniqueness pairs},
  author={Kulikov, Aleksei and Nazarov, Fedor and Sodin, Mikhail},
  journal={arXiv:2306.14013},
  year={2023}
}

@article{ramos2022fourier,
  title={Fourier uniqueness pairs of powers of integers},
  author={Ramos, Jo{\~a}o PG and Sousa, Mateus},
  journal={J. Eur. Math. Soc.(JEMS)},
  volume={24},
  pages={4327--4351},
  year={2022}
}

@article{ehler2024abstract,
  title={An abstract approach to {Marcinkiewicz-Zygmund} inequalities for approximation and quadrature in modulation spaces},
  author={Ehler, Martin and Gr{\"o}chenig, Karlheinz},
  journal={Mathematics of Computation},
  volume={93},
  number={350},
  pages={2885--2919},
  year={2024}
}

@article{Slepian1961,
  title = {Prolate Spheroidal Wave Functions,  {Fourier} Analysis and Uncertainty - {I}},
  volume = {40},
  ISSN = {0005-8580},
  DOI = {10.1002/j.1538-7305.1961.tb03976.x},
  number = {1},
  journal = {Bell System Technical Journal},
  publisher = {Institute of Electrical and Electronics Engineers (IEEE)},
  author = {Slepian,  D. and Pollak,  H. O.},
  year = {1961},
  pages = {43–63}
}

@article{Slepian1961b,
  title = {Prolate Spheroidal Wave Functions,  {Fourier} Analysis and Uncertainty - {II}},
  volume = {40},
  ISSN = {0005-8580},
  DOI = {10.1002/j.1538-7305.1961.tb03976.x},
  number = {1},
  journal = {Bell System Technical Journal},
  publisher = {Institute of Electrical and Electronics Engineers (IEEE)},
  author = {Slepian,  D. and Pollak,  H. O.},
  year = {1961},
  pages = {65–84}
}

@article{Landau1962,
  title = {Prolate Spheroidal Wave Functions,  {Fourier} Analysis and Uncertainty-III: The Dimension of the Space of Essentially Time- and Band-Limited Signals},
  volume = {41},
  ISSN = {0005-8580},
  DOI = {10.1002/j.1538-7305.1962.tb03279.x},
  number = {4},
  journal = {Bell System Technical Journal},
  publisher = {Institute of Electrical and Electronics Engineers (IEEE)},
  author = {Landau,  H. J. and Pollak,  H. O.},
  year = {1962},
  pages = {1295–1336}
}

@article{Kulikov2021,
  title = {Fourier Interpolation and Time-Frequency Localization},
  volume = {27},
  ISSN = {1531-5851},
  DOI = {10.1007/s00041-021-09861-y},
  number = {3},
  journal = {Journal of Fourier Analysis and Applications},
  publisher = {Springer Science and Business Media LLC},
  author = {Kulikov,  Aleksei},
  year = {2021},
}

@article{Aronszajn1950,
  title = {Theory of reproducing kernels},
  volume = {68},
  ISSN = {1088-6850},
  DOI = {10.1090/s0002-9947-1950-0051437-7},
  number = {3},
  journal = {Transactions of the American Mathematical Society},
  publisher = {American Mathematical Society (AMS)},
  author = {Aronszajn,  N.},
  year = {1950},
  pages = {337–404}
}

@book{Wahba1999,
author = {Whaba, G.},
title = {Support vector machines, reproducing kernel {Hilbert} spaces and the
randomized GACV},
publisher = {Advances in Kernel Methods, MIT Press},
year = {1999}
}

@book{Christmann2008,
author = {Christmann, A. and Steinwart, I.},
title = {Support Vector Machines},
edition = {2},
publisher = {Information Science and Statistics, Springer},
year = {2008}
}

@article{Szehr2020,
  title = {Interpolation without commutants},
  volume = {84},
  ISSN = {1841-7744},
  DOI = {10.7900/jot.2019may21.2264},
  number = {1},
  journal = {Journal of Operator Theory},
  publisher = {Theta Foundation},
  author = {Szehr,  Oleg and Zarouf,  Rachid},
  year = {2020},
  pages = {239–256}
}

@inproceedings{NIPS2012_eb160de1,
 author = {Dinuzzo, Francesco and Sch\"{o}lkopf, Bernhard},
 booktitle = {Advances in Neural Information Processing Systems},
 editor = {F. Pereira and C.J. Burges and L. Bottou and K.Q. Weinberger},
 pages = {},
 publisher = {Curran Associates, Inc.},
 title = {The representer theorem for {Hilbert} spaces: a necessary and sufficient condition},
 volume = {25},
 year = {2012}
}

@article{Zhang2009RKBS,
  title = {Reproducing Kernel {Banach} Spaces for Machine Learning},
  volume = {10},
  journal = {Journal of Machine Learning Research},
  author = {Zhang, Haizhang and Xu, Yuesheng},
  year = {2009},
  pages = {2741-2775}
}

\end{document}